\documentclass[runningheads]{llncs}
\usepackage[T1]{fontenc}
\usepackage{cite}
\usepackage{graphicx}
\usepackage{hyperref,wrapfig}
\usepackage{amsfonts,pifont}
\usepackage{amsmath}
\usepackage{amssymb}
\usepackage[epsilon,altpo]{backnaur}
\usepackage{tikz,diagbox}
\usetikzlibrary{matrix,positioning,fit,calc,arrows}
\usepackage[ruled,linesnumbered,vlined,noend]{algorithm2e}
\usepackage{soul}

\SetKwFunction{TermNorm}{TermNorm}
\SetKwFunction{AffConst}{AffConst}
\SetKwFunction{Verifier}{Verifier}

\usepackage{xcolor}
\usepackage{multirow}
\usepackage{semantic}
\usepackage{booktabs}
\usepackage{bbding}
\SetKwProg{Fn}{Function}{ is}{end}
\SetKwProg{Def}{Define}{ as}{end}
\SetKw{Break}{break}
\SetKw{Continue}{continue}
\DontPrintSemicolon

\newcommand{\tool}{\sf FISCHER}
\newcommand{\masked}{\mathfrak m}
\newcommand{\origin}{\mathfrak o}

\newcommand{\myvec}[1]{{\vec{\bf #1}}}
	
\definecolor{officegreen}{rgb}{0.0, 0.5, 0.0}

\begin{document}

\title{Automated Verification of Correctness for Masked Arithmetic Programs\thanks{This work is supported by the National Natural Science Foundation of China (62072309), CAS Project for Young Scientists in Basic
Research (YSBR-040), ISCAS New Cultivation Project (ISCAS-PYFX-202201),
an oversea grant from the State Key Laboratory of Novel Software Technology, Nanjing University (KFKT2022A03), and Birkbeck BEI School Project
(EFFECT).}}

\author{Mingyang Liu\inst{1} \and
Fu Song\inst{1,2,3}\textsuperscript{(\Envelope)}  \and
Taolue Chen\inst{4}}

\authorrunning{Liu et al.}

\institute{ShanghaiTech University,  Shanghai 201210, China  
\and
Institute of Software, Chinese Academy of Sciences \& University of Chinese Academy of
Sciences, Beijing 100190, China  \and
Automotive Software Innovation Center, Chongqing 400000, China \and
Birkbeck, University of London, London WC1E 7HX, United Kingdom} 

\maketitle              


\begin{abstract}
Masking is a widely-used effective countermeasure against power side-channel attacks for implementing cryptographic algorithms. Surprisingly, few formal verification techniques have addressed a fundamental question, i.e., whether the masked program and the original (unmasked) cryptographic algorithm are functional equivalent. In this paper, we study this problem for masked arithmetic programs over Galois fields of characteristic 2. We propose an automated approach based on term rewriting, aided by random testing and SMT solving. The overall approach is sound, and complete under certain conditions which do meet in practice. We implement the approach as a new tool {\tool} and carry out extensive experiments on various benchmarks. The results confirm the effectiveness, efficiency and scalability of our approach. Almost all the benchmarks can be proved for the first time by the term rewriting system solely. In particular, {\tool} detects a new flaw in a masked implementation published in EUROCRYPT 2017.
%
\end{abstract}
%
%

\section{Introduction}

Power side-channel attacks~\cite{kocher1996timing} can infer secrecy by statistically analyzing the
power consumption during the execution of cryptographic programs. The victims include
implementations of almost all major cryptographic algorithms, e.g., DES~\cite{kocher1999differential}, AES~\cite{ProuffRB09}, RSA~\cite{GP99}, Elliptic curve cryptography~\cite{OrsOP03,LFK18a}
and post-quantum cryptography~\cite{RaviRCB19,SchambergerRSW20}. 
To mitigate the threat, cryptographic algorithms are often implemented via \emph{masking}~\cite{ishai2003private}, which divides each secret value into $(d+1)$ shares by randomization, where $d$ is a given masking order.
However, it is error-prone to implement secure and correct masked implementations for non-linear functions (e.g., finite-field multiplication, module addition and S-Box), which are prevalent in  cryptography. Indeed, published implementations of AES S-Box that have been proved secure via paper-and-pencil~\cite{rivain2010provably,carlet2012higher,KHL11} were later shown to be vulnerable to power side-channels when 
$d$ is no less than $4$~\cite{CPRR13}.

While numerous formal verification techniques have been proposed to prove resistance of masked cryptographic programs against power side-channel attacks (e.g.,~\cite{EWS14b,zhang2018sc,BGIKMW18,GaoXZSC19,BBDFGSZ16,GaoZSW19,GaoXSC21,GaoXSZSC22}),
one fundamental question which is largely left open is the (functional) correctness of the masked cryptographic programs,
i.e., whether a masked program and the original (unmasked) cryptographic algorithm are actually functional equivalent.
%
It is conceivable to apply general-purpose program verifiers to masked cryptographic programs. Constraint-solving based approaches are available,
for instance, Boogie~\cite{barnett2005boogie} generates constraints via weakest precondition reasoning which then invokes SMT solvers;
SeaHorn~\cite{gurfinkel2015seahorn} and CPAChecker~\cite{beyer2011cpachecker} adopt model checking by utilizing SMT or CHC solvers.
More recent work (e.g., CryptoLine~\cite{tsai2017certified,PolyakovTWY18,LiuSTWY19,FuLSTWY19}) 
resorts to computer algebra, e.g., to reduce the problem to the ideal membership problem. 
The main challenge of applying these  techniques to masked cryptographic programs lies in the presence of finite-field multiplication, affine transformations and bitwise exclusive-OR (XOR). 
For instance, 
finite-field multiplication is not natively supported by the current SMT or CHC solvers,
and 
the increasing number of bitwise XOR operations causes the infamous state-explosion problem.
Moreover, to the best of our knowledge, current computer algebra systems do not 
provide the full support required by verification of masked cryptographic programs. 




\noindent{\bf Contributions.}
We propose a novel, term rewriting based approach to efficiently check
whether a masked program and the original (unmasked) cryptographic algorithm  (over Galois fields of characteristic 2) are functional equivalent.
Namely, we provide a term rewriting system (TRS) which 
can handle affine transformations, bitwise XOR, and finite-field multiplication.
The verification problem is reduced to checking whether a term can be rewritten to normal form $0$.
This approach is sound, i.e., once we obtain $0$, we can claim functional equivalence.
In case the TRS reduces to a normal form which is different from $0$, most likely they are \emph{not} functional equivalent, but 
a false positive is possible. We further resort to random testing and SMT solving by directly analyzing the obtained normal form.
As a result, it turns out that the overall approach is 
complete 
if no uninterpreted functions are involved in the normal form.

We implement our approach as a new tool {\tool} (\textbf{F}unctional\textbf{I}ty of ma\textbf{S}ked \textbf{C}ryptograp\textbf{H}ic program verifi\textbf{ER}), based on the LLVM framework~\cite{lattner2004llvm}.
We conduct extensive experiments on various masked cryptographic program benchmarks. 
The results show that our term rewriting system solely is able to prove almost all the benchmarks.
{\tool} is also considerably more efficient than
the general-purpose verifiers SMACK~\cite{rakamaric2014smack}, SeaHorn, CPAChecker, and  Symbiotic~\cite{chalupa2021symbiotic},
cryptography-specific verifier CryptoLine,  
as well as a straightforward approach that directly reduces the verification task to SMT solving.
For instance, our approach is able to handle masked implementations of finite-field multiplication with masking orders up to $100$  in less than 153 seconds,
while none of the compared approaches can handle masking order of $3$ in 20 minutes.

In particular, for the first time we detect a flaw in a masked implementation of finite-field multiplication published in EUROCRYPT 2017~\cite{barthe2017parallel}.
The flaw is tricky, as it only occurs for the masking order $d\equiv 1 \mod 4$.\footnote{This flaw has been confirmed by an author of~\cite{barthe2017parallel}.} This finding highlights the importance of the correctness verification of masked programs,
which has been largely overlooked, but of which our work provides an effective solution.

Our main contributions can be summarized as follows.
\begin{itemize}
    \item We propose a term rewriting system for automatically proving the functional correctness of
    masked cryptographic programs;
    \item We implement a tool {\tool} by synergistically integrating the term rewriting based approach, random testing and SMT solving; 
    \item We conduct extensive experiments, confirming the effectiveness, efficiency, scalability and  applicability of our approach.
\end{itemize}

\noindent
{\bf Related Work}.
Program verification has been extensively studied for decades. Here we mainly focus on their application in cryptographic programs, for which 
some general-purpose program verifiers have been adopted. 
Early work~\cite{almeida2017jasmin} uses Boogie~\cite{barnett2005boogie}.
HACL*~\cite{ZinzindohoueBPB17} uses F*~\cite{AhmanHMMPPRS17} which verifies programs by a combination of SMT solving and interactive  proof assistants.
Vale~\cite{bond2017vale} 
uses F* and  Dafny~\cite{leino2010dafny} where Dafny
harnesses Boogie for verification. 
Cryptol~\cite{tomb2016automated} checks equivalence between machine-readable cryptographic specifications and real-world implementations via SMT solving.
As mentioned before, computer algebra systems (CAS) have also been used 
for verifying cryptographic programs and arithmetic circuits, by reducing to the ideal membership problem together with 
SAT/SMT solving. Typical work includes CryptoLine and AMulet~\cite{kaufmann2019verifying,KaufmannB21}.
However, as shown in Section~\ref{sec:comparsion}, neither general-purpose verifiers (SMACK with Boogie and Corral, SeaHorn,
CPAChecker and Symbiotic) nor the CAS-based verifier CryptoLine
is sufficiently powerful to verify masked cryptographic programs.
%
Interactive proof assistants (possibly coupled with SMT solvers) have also been used to verify unmasked cryptographic programs
(e.g.,~\cite{BartheDGKSS13,Affeldt13,Appel15,MyreenC13,MyreenG07,chen2014verifying,ErbsenPGSC19}).
Compared to them, our approach is highly automatic, which is more acceptable and easier to use for general software developers. 

\smallskip
\noindent{\bf Outline}.
Section~\ref{sec:prel} recaps preliminaries. 
Section~\ref{sec:lang} presents a language on which the cryptographic program is formalized.
Section~\ref{sec:MotivatingandOverview} gives an example and an overview of our approach.
Section~\ref{sec:model} and Section~\ref{sec:algs} introduce the term rewriting system
and verification algorithms.
Section~\ref{sec:exp} reports experimental results.
We conclude in Section~\ref{sec:concl}. 
The source code of our tool and benchmarks are available at
\url{https://github.com/S3L-official/FISCHER}.

\section{Preliminaries}\label{sec:prel}

For two integers $l,u$ with $l\leq u$, $[l,u]$ denotes the set of integers $\{l,l+1, \cdots,u\}$.

\smallskip
\noindent{\bf Galois Field}. A \emph{Galois field} $\mathbb{GF}(p^n)$ comprises polynomials $a_{n-1}  X^{n-1}+\cdots + a_1  X^1+ a_0 $ over $\mathbb{Z}_p=[0,p-1]$, where $p$ is a prime number, $n$ is a positive integer, and $a_i\in \mathbb{Z}_p$. (Here $p$ is the \emph{characteristic} of the field, and
$p^n$ is the \emph{order} of the field.)
Symmetric cryptography (e.g., DES~\cite{nist1999des}, AES~\cite{daemen1999aes}, 
SKINNY~\cite{beierle2016skinny}, PRESENT~\cite{bogdanov2007present})
and bitsliced implementations of asymmetric cryptography (e.g.,~\cite{bronchain2022bitslicing}) intensively uses $\mathbb{GF}(2^n)$. 
Throughout the paper, $\mathbb{F}$ denotes the Galois field $\mathbb{GF}(2^n)$ for a fixed $n$, and $\oplus$ and $\otimes$ denote the addition and multiplication on $\mathbb{F}$, respectively.
Recall that $\mathbb{GF}(2^n)$ can be constructed from the quotient ring of the polynomial ring $\mathbb{GF}(2)[X]$
with respect to the ideal generated by an irreducible polynomial $P$ of degree $n$.
Hence, multiplication is the product of two polynomials
modulo $P$ in $\mathbb{GF}(2)[X]$ and addition is bitwise exclusive-OR (XOR) over the binary representation of polynomials.
For example, AES uses $\mathbb{GF}(256)=\mathbb{GF}(2)[X]/(X^8+X^4+X^3+X+1)$. Here $n=8$ and $P=X^8+X^4+X^3+X+1$.

\smallskip
\noindent{\bf Higher-Order Masking}.
To achieve order-$d$ security against power side-channel attacks under certain leakage models,
masking is usually used~\cite{shamir1979share,ishai2003private}. 
Essentially, masking partitions each secret value into (usually $d+1$) shares so that knowing at most $d$ shares cannot infer
any information of the  secret value, called \emph{order-$d$ masking}.
%
In Boolean masking,
a value $a\in \mathbb{F}$ is divided into shares $a_0, a_1, \ldots, a_d\in\mathbb{F}$
such that $a_0 \oplus a_1 \oplus \ldots \oplus a_d = a$. Typically,
$a_1,  \ldots, a_d$ are random values and $a_0= a\oplus a_1 \oplus \ldots \oplus a_d$.
The tuple $(a_0, a_1, \ldots, a_d)$, denoted by $\myvec{a}$, is called an \emph{encoding} of $a$.
We write $\bigoplus_{i \in [0,d]} {\myvec{a}}_i$ (or simply $\bigoplus \myvec{a}$) for $a_0 \oplus a_1 \oplus \ldots \oplus a_d$.
Additive masking can be defined similarly to Boolean masking, where $\oplus$ is replaced by the module arithmetic addition operator.
In this work, we focus on Boolean masking as the XOR operation is more efficient to implement.

To implement a masked program, for each operation in the cryptographic algorithm,  
a corresponding operation on shares is required. 
As we will see later, when the operation is affine (i.e. the operation $f$ satisfies $f(x \oplus y) = f(x) \oplus f(y)\oplus c$ for some constant $c$),  
the corresponding operation is simply  
to apply the original operation on each share $a_i$ in the encoding $(a_0, a_1, \ldots, a_d)$.
However, for non-affine operations (e.g., multiplication and addition), it is a very difficult
task and error-prone~\cite{CPRR13}.
Ishai et al.~\cite{ishai2003private}   
proposed the first
masked implementation of multiplication, but 
limited to the domain $\mathbb{GF}(2)$ \emph{only}. The number of the required random values and operations
is not optimal and is known to be vulnerable in the presence of glitches because the electric signals propagate at different speeds in
the combinatorial paths of hardware circuits. Thus, various follow-up papers proposed ways to implement higher-order masking
for the domain $\mathbb{GF}(2^n)$ and/or optimizing the computational complexity, e.g.,~\cite{rivain2010provably,belaid2016randomness,barthe2017parallel,gross2017reconciling,cassiers2020trivially},
all of which are referred to as ISW scheme in this paper.
In another research direction, new glitch-resistant Boolean masking schemes have been proposed, e.g.,
Hardware Private Circuits (HPC1 \& HPC2)~\cite{cassiers2020hardware}, Domain-oriented Masking (DOM)~\cite{gross2016domain} and Consolidating Masking Schemes (CMS)~\cite{reparaz2015consolidating}.
In this work, we are interested in automatically proving the correctness of the masked programs.

\section{The Core Language}\label{sec:lang}
In this section, we first present the core language \textsf{MSL},  given in Figure~\ref{fig:msl-syntax}, based on which the verification problem is formalized.

\begin{figure}[t]
	\centering
  \includegraphics[width=1\textwidth]{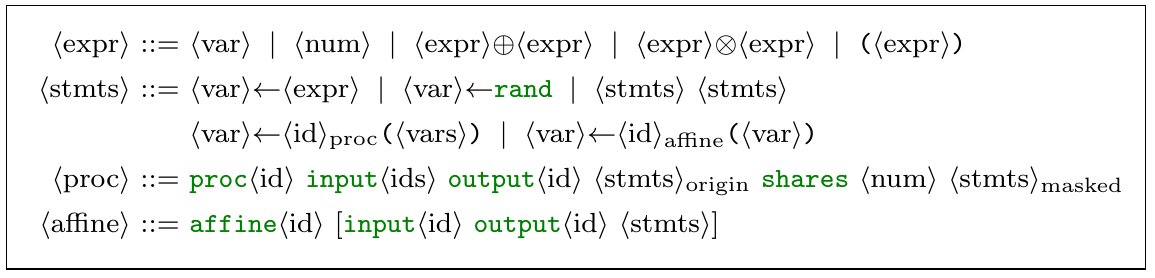}
\caption{Syntax of \textsf{MSL} in Backus-Naur form} 
\label{fig:msl-syntax}
\end{figure}

A program $\mathcal{P}$ in \textsf{MSL} is given by a sequence of procedure definitions and affine transformation definitions/declarations.
A procedure definition starts with the keyword \textcolor{officegreen}{proc}, followed by a procedure name, a list of input parameters, an output and its body.
The procedure body has two blocks of statements, separated by a special statement \textcolor{officegreen}{shares} $d+1$, where
$d$ is the masking order.
The first block $\langle$stmts$\rangle_\text{origin}$, called the \emph{original block}, implements its original functionality on the input parameters without masking.
The second block $\langle$stmts$\rangle_\text{masked}$, called the \emph{masked block}, is a masked implementation of the original block over the input encodings $\myvec{x}$ of the input parameters $x$. The input parameters and output $x$, declared using the keywords \textcolor{officegreen}{input}
and \textcolor{officegreen}{output} respectively, are scalar variables in the original block, but are treated as the corresponding
encodings (i.e., tuples) $\myvec{x}$ in the masked block.
For example, \textcolor{officegreen}{input} $x$ declares the scalar variable $x$ as the input of the original block,
while it implicitly declares an encoding $\myvec{x}=(x_0, x_1, \ldots, x_d)$ as the input of the masked block with \textcolor{officegreen}{shares} $d+1$.

We distinguish affine transformation definitions and declarations. The former starts with the keyword \textcolor{officegreen}{affine}, followed by a name  $f$, an input, an output and its body. It is expected that the affine property $\forall x, y\in \mathbb{F}. f(x \oplus y) = f(x) \oplus f(y) \oplus c$ holds for some affine constant $c\in\mathbb{F}$.
(Note that the constant $c$ is not explicitly provided in the program, but can be derived, cf.\ Section~\ref{sect:const}.)
The transformation $f$ is \textit{linear} if its affine constant $c$ is $0$.
In contrast, an affine transformation declaration $f$ simply declares
a transformation. As a result, it can only be used to declare
a linear one (i.e., $c$ must be 0),  
which is treated as an uninterpreted function.
Note that non-linear affine transformation declarations
can be achieved by declaring linear affine transformations and affine transformation definitions.
Affine transformation here serves as an abstraction to capture complicated operations (e.g.,  shift, rotation and bitwise Boolean operations) and can accelerate verification by expressing operations as uninterpreted functions. In practice, a majority of cryptographic algorithms (in symmetric cryptography) can be represented by a composition of S-box, XOR and linear transformation only.

Masking an affine transformation can simply mask  an input encoding in a share-wise way, namely,
the masked version of the affine transformation $f(a)$ is 
\begin{center}
$f(a_0 \oplus a_1 \oplus \ldots \oplus a_d)=
\left\{
  \begin{array}{ll}
    f(a_0) \oplus f(a_1) \oplus \ldots \oplus f(a_d), & \hbox{if $d$ is even;} \\
    f(a_0) \oplus f(a_1) \oplus \ldots \oplus f(a_d)\oplus c, & \hbox{if $d$ is odd.}
  \end{array}
\right.$
\end{center}
This is default, so 
affine transformation definition only contains the original block  but no masked block. 

A statement is either an assignment or a function call.
\textsf{MSL} features two types of assignments which are either of the form $x \leftarrow e$ defined as usual or of the form $r \leftarrow$\textcolor{officegreen}{rand} which assigns a uniformly sampled value from the domain  $\mathbb{F}$ to the variable $r$. As a result, $r$ should be read as a random variable. We assume that each random variable is defined only once.
%
We note that the actual parameters and output are scalar if the procedure is invoked in an original block
while they are the corresponding encodings if it is invoked in a masked block.


\textsf{MSL} is the core language of our tool. In practice, to be more user-friendly, our tool also accepts C programs with
conditional branches and loops, both of which should be statically determinized (e.g., loops are bound and can be unrolled;
the branching of conditionals can also be fixed after loop unrolling). Furthermore, we assume there is no recursion and dynamic memory allocation.
These restrictions are sufficient for most symmetric cryptography and bitsliced implementations of public-key cryptography, which mostly have simple control graphs and memory aliases.


\smallskip
\noindent {\bf Problem formalization.}
Fix a program $\mathcal{P}$ with all the procedures using order-$d$ masking.
We denote by $\mathcal{P}_{\origin}$ (resp.\ $\mathcal{P}_{\masked}$) the program $\mathcal{P}$ where all the masked (resp.\ original) blocks  are omitted.
For each procedure $f$, the procedures $f_{\origin}$ and $f_{\masked}$ are defined accordingly.

\begin{definition}\label{def:problem}
Given a procedure $f$ of $\mathcal{P}$ with $m$ input parameters, $f_{\masked}$ and $f_{\origin}$ are \emph{functional equivalent}, denoted by $f_{\masked}\cong f_{\origin}$, if
the following statement holds:
\begin{center}
 \flushleft{$\forall a^1, \cdots, a^m, r_1,\cdots,r_h \in \mathbb{F}, \forall\myvec{a}^1,   \cdots, \myvec{a}^m \in \mathbb{F}^{d+1}.$}\\ \hfill
	$\big(
	\bigwedge_{i \in [1,m]}\ a^i = \bigoplus_{j \in [0,d]} {\myvec{a}}^i_j
	\big)
	\rightarrow
	\big(f_{\origin}(a^1,\cdots, a^m) = \bigoplus_{i \in [0,d]} f_{\masked}(\myvec{a}^1, \cdots, \myvec{a}^m)_i\big)$
\end{center}
\flushleft{where $r_1,\cdots,r_h$ are all the random variables used in $f_{\masked}$.}
\end{definition}
Note that although the procedure $f_{\masked}$ is randomized (i.e., the output encoding $f_{\masked}(\myvec{a}^1, \cdots, \myvec{a}^m_i)$ is technically a random variable), 
for functional equivalence we consider a stronger notion, viz., to require that $f_{\masked}$ and $f_{\origin}$ are equivalent under any values in the support of the random variables  $r_1,\cdots,r_h$. Thus, $r_1,\cdots,r_h$ are universally quantified  in Definition~\ref{def:problem}.


The verification problem is to check if $f_{\masked}\cong f_{\origin}$ for a given procedure $f$ where $\bigwedge_{i \in [1,m]}\ a^i = \bigoplus_{j \in [0,d]} {\myvec{a}}^i_j$ and $f_{\origin}(a^1,\cdots, a^m) = \bigoplus_{i \in [0,d]} f_{\masked}(\myvec{a}^1, \cdots, \myvec{a}^m)_i$
are regarded as pre- and post-conditions, respectively.
Thus, we assume  the unmasked procedures themselves
are correct (which can be verified by, e.g., CryptoLine). Our focus is on whether the masked counterparts are functional equivalent to them.

\begin{figure}[t]
\centering
  \includegraphics[width=.95\textwidth]{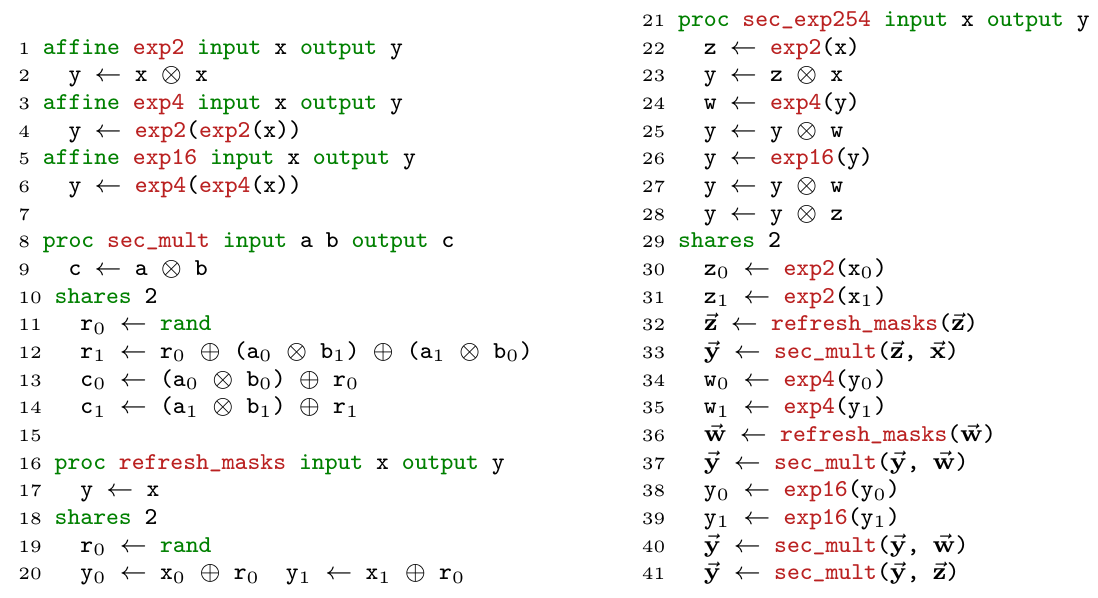}\vspace*{-2mm}
\caption{Motivating example, where $\myvec{x}$ denotes $(x_0,x_1)$.}
\label{fig:sec-exp254}\vspace*{-2mm}
\end{figure}


\section{Overview of the approach}\label{sec:MotivatingandOverview}
In this section, we first present a motivating example given in Figure~\ref{fig:sec-exp254}, which computes the multiplicative inverse in $\mathbb{GF}(2^8)$ for the AES S-Box~\cite{rivain2010provably} using
first-order Boolean masking.
It consists of three affine transformation definitions
and two procedure definitions. For a given input $x$, \texttt{exp2}$(x)$ outputs $x^2$,
\texttt{exp4}$(x)$ outputs $x^4$ and \texttt{exp16}$(x)$ outputs $x^{16}$.
Obviously, these three affine transformations are indeed linear.

Procedure \texttt{sec\_mult}$_{\origin}(a,b)$ outputs $a\otimes b$.
Its masked version  \texttt{sec\_mult}$_{\masked}(\myvec{a},\myvec{b})$ computes the encoding $\myvec{c}=(c_0,c_1)$ over
the encodings $\myvec{a}=(a_0,a_1)$ and $\myvec{b}=(b_0,b_1)$.
Clearly, it is desired that
$c_0\oplus c_1=a\otimes b $ if $a_0\oplus a_1=a$ and $b_0\oplus b_1=b$.
Procedure \texttt{refresh\_masks}$_{\origin}(x)$ is the identity function while its masked version
\texttt{refresh\_masks}$_{\masked}(\myvec{x})$ re-masks the encoding $\myvec{x}$
 using a random variable $r_0$. Thus, it is desired that $y_0\oplus y_1=x$ if $x=x_0\oplus x_1$.
Procedure \texttt{sec\_exp254}$_{\origin}(x)$ computes the multiplicative inverse $x^{254}$ of $x$ in $\mathbb{GF}(2^8)$.
Its masked version \texttt{sec\_exp254}$_{\masked}(\myvec{x})$ computes the encoding $\myvec{y}=(y_0,y_1)$
where \texttt{refresh\_masks}$_{\masked}$ is invoked to avoid power side-channel leakage.
Thus, it is desired that
$y_0\oplus y_1=x^{254}$ if $x_0\oplus x_1=x$.
In summary, it is required to prove
$\texttt{sec\_mult}_{\masked}\cong \texttt{sec\_mult}_{\origin}$, $\texttt{refresh\_masks}_{\masked}\cong \texttt{refresh\_masks}_{\origin}$
and $\texttt{sec\_exp254}_{\masked}\cong \texttt{sec\_exp254}_{\origin}$.


\subsection{Our Approach}
An overview of {\tool} is shown in Figure~\ref{fig:overview}.
The input  program is expected to follow the syntax of \textsf{MSL} but in C language. Moreover, the pre-conditions and post-conditions of the verification problem are expressed by \textcolor{officegreen}{assume} and \textcolor{officegreen}{assert} statements in the masked
procedure, respectively. Recall that the input program can contain conditional branches and loops when are statically determinized. Furthermore, affine transformations can use other common operations (e.g., shift, rotation and bitwise Boolean operations) besides the addition $\oplus$ and multiplication  $\otimes$ on the underlying field $\mathbb{F}$.
{\tool} leverages the LLVM framework
to obtain the LLVM intermediate representation (IR)
and call graph, where all the procedure calls are inlined.
It then invokes \emph{Affine Constant Computing} to iteratively compute the affine constants for affine transformations according to the call graph,
and  \emph{Functional Equivalence Checking} to check
functional equivalence, both of which
rely on the underpinning engines, viz., \emph{Symbolic Execution} (refer to symbolic computation without path constraint solving in this work), \emph{Term Rewriting} and  \emph{SMT-based Solving}.

\begin{figure}[t]
  \centering
  \includegraphics[width=.8\textwidth]{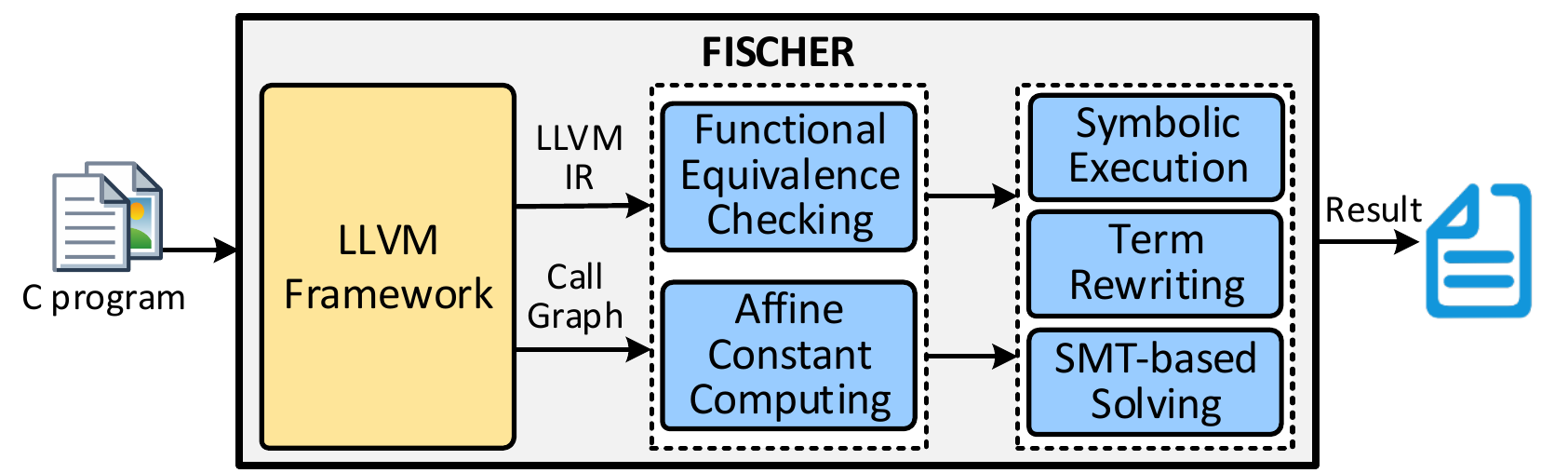}
  \caption{Overview of {\tool}.}\label{fig:overview}\vspace*{-4mm}
\end{figure}

We apply intra-procedural symbolic execution to compute the symbolic outputs of the procedures and transformations, i.e.,
 expressions in terms of inputs, random variables and affine transformations.
The symbolic outputs are treated as terms based on which both the problems of functional equivalence checking and affine constant computing
are solved by rewriting to their normal forms (i.e., sums of monomials w.r.t. a total order).
The analysis result is often conclusive from normal forms. 
In case it is inconclusive, we iteratively inline affine transformations when their definitions are available until
either the analysis result is conclusive or no more affine transformations can be inlined.
If the analysis result is still inconclusive, to reduce false positives, we apply random testing and accurate (but computationally expansive) SMT solving to the normal forms instead of the original terms.
We remark that the term rewriting system solely can prove almost all the benchmarks in our experiments.

Consider the motivating example. To find the  constant $c\in\mathbb{F}$ of {\tt exp2} 
such that the property
$\forall x, y\in \mathbb{F}. {\tt exp2}(x \oplus y) = {\tt exp2}(x) \oplus {\tt exp2}(y) \oplus c$ holds,
by applying symbolic execution, ${\tt exp2}(x)$ is expressed as the term $x\otimes x$. Thus, the property
is reformulated as
$(x \oplus y)\otimes (x \oplus y) = (x\otimes x) \oplus (y\otimes y) \oplus c$,
from which we can deduce that  the desired affine constant $c$ is equivalent to the term $((x \oplus y)\otimes (x \oplus y)) \oplus (x\otimes x) \oplus (y\otimes y)$.
Our TRS
will reduce the term as follows:\\ 
{\footnotesize\selectfont\begin{tabular}{lr}
 $\quad\underline{((x \oplus y)\otimes (x \oplus y))} \oplus (x\otimes x) \oplus (y\otimes y)$ & \quad Distributive Law\\
  $=\underline{(x\otimes (x \oplus y))} \oplus \underline{(y\otimes (x \oplus y))} \oplus (x\otimes x) \oplus (y\otimes y)$ & \quad Distributive Law\\
  $=(x\otimes x)\oplus (x\otimes y) \oplus \underline{(y\otimes x)} \oplus (y \otimes y) \oplus (x\otimes x) \oplus (y\otimes y)$ & \quad Commutative Law\\
  $=(x\otimes x)\oplus (x\otimes y) \oplus (x\otimes y) \oplus (y \otimes y) \oplus \underline{(x\otimes x)} \oplus (y\otimes y)$ &  \quad Commutative Law\\
  $=\underline{(x\otimes x)\oplus (x\otimes x)}\oplus \underline{(x\otimes y) \oplus (x\otimes y)} \oplus \underline{(y \otimes y)  \oplus (y\otimes y)}=0$ &  \quad Zero Law of XOR\\
\end{tabular}}

For the transformation {\tt exp4}$(x)$, by applying  symbolic execution, it can be expressed as the term ${\tt exp2}({\tt exp2}(x))$.
To find the constant $c\in\mathbb{F}$ to satisfy $\forall x, y\in \mathbb{F}. {\tt exp4}(x \oplus y) = {\tt exp4}(x) \oplus {\tt exp4}(y) \oplus c$,
we compute the term ${\tt exp2}({\tt exp2}(x \oplus y))\oplus {\tt exp2}({\tt exp2}(x))\oplus {\tt exp2}({\tt exp2}(y))$.
By applying our TRS, we have:

{\footnotesize\selectfont\begin{tabular}{l}
 $\quad \underline{{\tt exp2}({\tt exp2}(x \oplus y))}\oplus {\tt exp2}({\tt exp2}(x))\oplus {\tt exp2}({\tt exp2}(y))$ \\
  $= \underline{{\tt exp2}({\tt exp2}(x) \oplus {\tt exp2}(y))}\oplus {\tt exp2}({\tt exp2}(x))\oplus {\tt exp2}({\tt exp2}(y))$\\
  $= {\tt exp2}({\tt exp2}(x)) \oplus {\tt exp2}({\tt exp2}(y))\oplus \underline{{\tt exp2}({\tt exp2}(x))}\oplus {\tt exp2}({\tt exp2}(y))$\\
  $= \underline{{\tt exp2}({\tt exp2}(x)) \oplus {\tt exp2}({\tt exp2}(x))}\oplus \underline{{\tt exp2}({\tt exp2}(y))\oplus {\tt exp2}({\tt exp2}(y))}=0$\\
\end{tabular}}

\noindent Clearly, the affine constant of ${\tt exp4}$ is $0$.
Similarly, we can deduce that the affine constant of the transformation {\tt exp16} is $0$ as well.

To prove $\texttt{sec\_mult}_{\origin}\cong\texttt{sec\_mult}_{\masked}$, by applying symbolic execution, we have that
 $\texttt{sec\_mult}_{\origin}(a,b)=a\otimes b$ and $\texttt{sec\_mult}_{\masked}(\myvec{a},\myvec{b})=\myvec{c}=(c_0,c_1)$,
where $c_0= (a_0\otimes b_0)\oplus r_0$ and $c_1=(a_1\otimes b_1)\oplus (r_0\oplus (a_0\otimes b_1)\oplus (a_1\otimes b_0))$.
Then, by Definition~\ref{def:problem}, it suffices to check
\begin{center}
	\flushleft{$\forall a,b,a_0,a_1,b_0,b_1,r_0 \in \mathbb{F}.
	\big(
	a = a_0\oplus a_1 \wedge  b = b_0\oplus b_1
	\big)\rightarrow$} \\ \hfill
	$\big(a\otimes b = ((a_0\otimes b_0)\oplus r_0)\oplus \big((a_1\otimes b_1)\oplus (r_0\oplus (a_0\otimes b_1)\oplus (a_1\otimes b_0))\big) \big)$.
\end{center}

\smallskip
\noindent Thus, we check the term $\big((a_0\oplus a_1) \otimes (b_0\oplus b_1)\big)\oplus ((a_0\otimes b_0)\oplus r_0)\oplus ((a_1\otimes b_1)\oplus (r_0\oplus (a_0\otimes b_1)\oplus (a_1\otimes b_0)))$
which is equivalent to $0$ iff $\texttt{sec\_mult}_{\origin}\cong\texttt{sec\_mult}_{\masked}$.
Our TRS is able to reduce the term to $0$.
Similarly, we represent the outputs of $\texttt{sec\_exp254}_{\origin}$ and $\texttt{sec\_exp254}_{\masked}$ as terms via symbolic execution, from which
the statement $\texttt{sec\_exp254}_{\origin}\cong\texttt{sec\_exp254}_{\masked}$ is also encoded as a term, which can be reduced to $0$ via our TRS without inlining any transformations.

\section{Term Rewriting System}\label{sec:model}
In this section, we first introduce some basic notations and then present our term rewriting system.

\begin{definition}
Given a program $\mathcal{P}$ over $\mathbb{F}$, a \emph{signature} $\Sigma_\mathcal{P}$ of $\mathcal{P}$ is a set of symbols $\mathbb{F}\cup\{\oplus, \otimes, f_1, \ldots, f_t\}$,
where $s\in \mathbb{F}$ with arity $0$ are all the constants in $\mathbb{F}$,
$\oplus$ and $\otimes$ with arity $2$ are addition and multiplication operators on $\mathbb{F}$,
and $f_1,\cdots, f_t$ with arity $1$ are affine transformations defined/declared in $\mathcal{P}$.
\end{definition}

For example, the signature of the motivating example is $\mathbb{F}\cup\{\oplus, \otimes, {\tt exp2},{\tt exp4},\\ {\tt exp16}\}$.
When it is clear from the context, the subscript $\mathcal{P}$ is dropped from $\Sigma_\mathcal{P}$.

\begin{definition}
Let $V$ be a set of variables  (assuming $\Sigma\cap V=\emptyset$), the set
$T[\Sigma,V]$ of \emph{$\Sigma$-terms} over $V$ is inductively defined as follows:
\begin{itemize}
  \item $\mathbb{F}\subseteq T[\Sigma,V]$ and $V\subseteq T[\Sigma,V]$ (i.e., every variable/constant is a $\Sigma$-term);
  \item $\tau \oplus \tau'\in T[\Sigma,V]$ and $\tau \otimes \tau'\in T[\Sigma,V]$ if $\tau,\tau'\in  T[\Sigma,V]$ (i.e., application of addition and multiplication operators to $\Sigma$-terms yield $\Sigma$-terms);
  \item $f_j(\tau)\in T[\Sigma,V]$   if $\tau\in  T[\Sigma,V]$ and $j\in[1,t]$ (i.e., application of affine transformations to $\Sigma$-terms yield $\Sigma$-terms).
\end{itemize}
We denote by $T_{\backslash\oplus}(\Sigma,V)$ the set of $\Sigma$-terms that do not use the operator $\oplus$.
\end{definition}

A $\Sigma$-term $\alpha\in T[\Sigma,V]$ is called a \emph{factor} if $\tau \in \mathbb{F}\cup V$ or $\tau=f_i(\tau')$ for some $i\in[1,t]$ such that $\tau'\in T_{\backslash\oplus}(\Sigma,V)$.
A \emph{monomial} is a product $\alpha_1\otimes\cdots \otimes \alpha_k$ of none-zero factors for $k\geq 1$.
We denote by $M[\Sigma,V]$ the set of monomials.
%
%
For instance, consider variables $x,y\in V$ and affine transformations $f_1,f_2\in \Sigma$. All $f_1(f_2(x))\otimes f_1(y)$, $f_1(2\otimes f_2(4\otimes x))$, $f_1(x\oplus y)$ and $f_1(f_2(x)) \oplus f_1(x)$ are $\Sigma$-terms,
both $f_1(f_2(x))\otimes f_1(y)$ and $f_1(2\otimes f_2(4\otimes x))$ are monomials, while neither $f_1(x\oplus y)$ nor  $f_1(f_2(x)) \oplus f_1(x)$ is a monomial.
For the sake of presentation, $\Sigma$-terms will be written as terms, and the operator $\otimes$ may be omitted, e.g., $\tau_1\tau_2$ denotes $\tau_1\otimes\tau_2$,
and $\tau^2$ denotes $\tau\otimes\tau$.

\begin{definition}
A polynomial is a sum  $\bigoplus_{i\in[1,t]} m_i$ of monomials  $m_1 \ldots m_t\in M[\Sigma, V]$. We use $P[\Sigma,V]$ to denote the set of polynomials.
\end{definition}

To simplify and normalize polynomials, we impose a total order on monomials 
and their factors.

\begin{definition}
Fix an arbitrary total order $\geq_s$ on $V\uplus \Sigma$.

For two factors $\alpha$ and $\alpha'$, the \emph{factor order} $\geq_l$ is defined such that $\alpha \geq_l \alpha'$ if one of the following conditions holds:
\begin{itemize}
  \item $\alpha,\alpha'\in \mathbb{F}\cup V$ and $\alpha \geq_s \alpha'$; 
  \item $\alpha=f(\tau)$ and $\alpha'=f'(\tau')$ such that $f \geq_s f'$ or ($f = f'$ and $\tau \geq_p \tau'$);
  \item  $\alpha=f(\tau)$ such that $f \geq_s \alpha'$  or $\alpha'=f(\tau)$ such that $\alpha \geq_s f$.
\end{itemize}
Given a monomial $m=\alpha_1\cdots  \alpha_k$, we write ${\sf sort}_{\geq_l}(\alpha_1,\cdots, \alpha_k)$ for the monomial which includes $\alpha_1, \cdots, \alpha_k$ as factors, but sorts them in descending order.

Given two monomials $m=\alpha_1\cdots  \alpha_k$ and $m'=\alpha_1'\cdots  \alpha_{k'}'$, the
 \emph{monomial order} $\geq_p$ is defined as the lexicographical order between ${\sf sort}_{\geq_l}(\alpha_1,\cdots, \alpha_k)$ and ${\sf sort}_{\geq_l}(\alpha_1',\cdots, \alpha_{k'}')$.
\end{definition}

Intuitively, the factor order $\geq_l$ follows the given order $\geq_s$ on $V\uplus \Sigma$,
where the factor order between two factors with the same affine transformation $f$ is determined by their parameters.
We note that if
${\sf sort}_{\geq_l}(\alpha_1',\cdots, \alpha_{k'}')$ is a prefix of  ${\sf sort}_{\geq_l}(\alpha_1,\cdots, \alpha_k)$,
we have: $\alpha_1\cdots  \alpha_k \geq_p \alpha_1'\cdots  \alpha_{k'}'$.
Furthermore, if $\alpha_1\cdots  \alpha_k \geq_p \alpha_1'\cdots  \alpha_{k'}'$ and
$\alpha_1'\cdots  \alpha_{k'}' \geq_p \alpha_1\cdots  \alpha_k$, then
${\sf sort}_{\geq_l}(\alpha_1',\cdots, \alpha_{k'}')={\sf sort}_{\geq_l}(\alpha_1,\cdots, \alpha_k)$.
We denote by $\alpha_1\cdots  \alpha_k >_p \alpha_1'\cdots  \alpha_{k'}'$
if $\alpha_1\cdots  \alpha_k \geq_p \alpha_1'\cdots  \alpha_{k'}'$ but ${\sf sort}_{\geq_l}(\alpha_1',\cdots, \alpha_{k'}')\neq {\sf sort}_{\geq_l}(\alpha_1,\cdots, \alpha_k)$.

\begin{proposition}
The monomial order $\geq_p$ is a total order on monomials.
\end{proposition}

\begin{definition}
Given a program $\mathcal{P}$, we define the corresponding term rewriting system (TRS) $\mathcal{R}$ as a tuple $(\Sigma, V, \geq_s, \Delta)$, where
$\Sigma$ is a signature of $\mathcal{P}$, $V$ is a set of variables of $\mathcal{P}$ (assuming $\Sigma\cap V=\emptyset$),
$\geq_s$ is a total order on $V\uplus \Sigma$, and $\Delta$ is the set of term rewriting rules
given below:
\vspace*{-4mm}
 \begin{center}
\scalebox{0.85}{\begin{tabular}{llll}
  \multicolumn{2}{l}{$\inference[R1]{(m_1', \cdots, m_k')={\tt sort}_{\geq_p} (m_1, \cdots, m_k)\neq (m_1, \cdots, m_k)} { m_1\oplus \cdots \oplus  m_k\mapsto m_1' \oplus \cdots \oplus  m_k'}$} &
  $\inference[R3]{} {\tau\oplus \tau\mapsto 0}$  & $\inference[R5]{}  {0 \tau\mapsto 0}$ \\ \specialrule{0em}{2pt}{2pt}
  \multicolumn{2}{l}{$\inference[R2]{(\alpha_1', \cdots, \alpha_k')={\tt sort}_{\geq_l} (\alpha_1, \cdots, \alpha_k)\neq (\alpha_1, \cdots, \alpha_k)} {\alpha_1 \cdots \alpha_k\mapsto \alpha_1' \cdots \alpha_k'}$} & $\inference[R4]{} {\tau 0\mapsto 0}$  &  $\inference[R6]{}{\tau\oplus 0\mapsto \tau}$   \\ \specialrule{0em}{2pt}{2pt}
\multicolumn{2}{l}{$\inference[R7]{}{0\oplus \tau\mapsto \tau}$ \qquad $\inference[R8]{} {\tau 1 \mapsto \tau}$  \qquad $\inference[R9]{} {1  \tau \mapsto \tau}$} &   \multicolumn{2}{l}{$\inference[R10]{} {(\tau_1\oplus \tau_2) \tau\mapsto (\tau_1 \tau)\oplus (\tau_2 \tau)}$ }\\ \specialrule{0em}{2pt}{2pt}
 $\inference[R11]{} {\tau(\tau_1\oplus \tau_2) \mapsto (\tau \tau_1)\oplus (\tau \tau_2)}~~$ &
 \multicolumn{2}{l}{$\inference[R12]{} {f(\tau_1\oplus \tau_2)\mapsto f(\tau_1)\oplus f(\tau_2)\oplus c}$}  & $\inference[R13]{} { f(0)\mapsto c}$    \\ \specialrule{0em}{2pt}{2pt}
\end{tabular}}
 \end{center}
\vspace*{-1mm}
where $m_1,m_1',\cdots,m_k,m_k'\in M[\Sigma, V]$, $\alpha_1,\alpha_2,\alpha_3$ are factors, $\tau,\tau_1,\tau_2\in T[\Sigma,V]$ are terms, $f\in\Sigma$ is an affine transformation
with affine constant $c$.
 \end{definition}

Intuitively,  rules R1 and R2 specify the commutativity of $\oplus$ and $\otimes$, respectively, by which
monomials and factors are sorted according to the orders $\geq_{p}$ and $\geq_{l}$, respectively.
Rule R3 specifies that 
$\oplus$ is essentially bitwise XOR.
Rules R4 and  R5 specify that $0$ is the multiplicative zero.
Rules R6 and  R7 (resp. R8 and  R9) specify that $0$ (resp. $1$) is additive (resp. multiplicative) identity.
Rules R10 and  R11 express the distributivity of $\otimes$ over $\oplus$.
Rule R12 expresses the affine property of an affine transformation while rule R13 is an instance of rule R12 via rules R3 and R5.

Given a TRS $\mathcal{R}=(\Sigma, V, \geq_s, \Delta)$ for a given program $\mathcal{P}$, a term $\tau\in T[\Sigma,V]$ can be rewritten to a term
$\tau'$, denoted by $\tau\Rightarrow \tau'$, if there is a rewriting rule $\tau_1\mapsto \tau_2$ such that 
 $\tau'$ is a term obtained from $\tau$ by replacing an occurrence of the sub-term $\tau_1$ with the sub-term $\tau_2$.
A term is in a \emph{normal form} if no rewriting rules can be applied.
A TRS is \emph{terminating} if all terms 
can be rewritten to a normal form 
after finitely many rewriting. We denote by $\tau\Rrightarrow \tau'$ with $\tau'$ being the normal form of $\tau$.

We show that any TRS $\mathcal{R}$ associated with a program $\mathcal{P}$ is terminating, and that any term will be rewritten to a normal form that is a polynomial, independent of the way of applying rules.

\begin{lemma}\label{lemma:normalform}
For every normal form $\tau\in T[\Sigma,V]$ of the TRS $\mathcal{R}$, the term $\tau$ must be a polynomial $m_1\oplus \cdots \oplus  m_k$  such that
(1) $\forall i\in[1,k-1]$, $m_i>_p m_{i+1}$, and (2) for every monomial $m_i=\alpha_1 \cdots \alpha_h$
and $\forall i\in[1,h-1]$, $\alpha_i \geq_l \alpha_{i+1}$.
\end{lemma}
\begin{proof}
Consider a normal form $\tau\in T[\Sigma,V]$. If $\tau$ is not a polynomial, then there must exist some monomial $m_i$ in which
the addition operator $\oplus$ is used. This means that either rule R$_{10}$ or R$_{11}$ is applicable to the term $\tau$ which contradicts the fact that $\tau$
is normal form.

Suppose $\tau$ is the polynomial $m_1\oplus \cdots \oplus  m_k$.
\begin{itemize}
  \item If there exists $i:1\leq i< k$ such that  $m_i>_p m_{i+1}$ does not hold, then either $m_i=m_{i+1}$ or $m_{i+1}>_p m_i$.
  If $m_i= m_{i+1}$, then rule R3 is applicable to the term $\tau$. If $m_{i+1}>_p m_i$, then rule R$_{1}$ is applicable to the term $\tau$.
  Thus, for every $1\leq i< k$, $m_i>_p m_{i+1}$.
  \item If there exist a monomial $m_i=\alpha_1 \cdots \alpha_h$
and $i:1\leq i< h$ such that $\alpha_i \geq_l \alpha_{i+1}$ does not hold, then $\alpha_{i+1} >_l \alpha_i$. This means that rule R2 is applicable to the term $\tau$.
Thus,  for every monomial $m_i=\alpha_1 \cdots \alpha_h$
and every $i:1\leq i< h$, $\alpha_i \geq_l \alpha_{i+1}$.\qed
\end{itemize}
\end{proof}

\begin{lemma}\label{lemma:terminating}
The TRS $\mathcal{R}=(\Sigma, V, \geq_s, \Delta)$ of a given program $\mathcal{P}$ is terminating.
\end{lemma}
\begin{proof}
Consider a term $\tau\in T[\Sigma,V]$. Let $\pi=\tau_1\Rightarrow \tau_2\Rightarrow\tau_3\Rightarrow\cdots \Rightarrow\tau_i\Rightarrow\cdots$ be a reduction
of the term $\tau$ by applying rewriting rules, i.e., $\tau=\tau_1$. We prove that the reduction $\pi$ is finite by showing that all the rewriting rules can be applied
finitely.

First, since rules R1 and R2  only sort the monomials and factors, respectively, while sorting always terminates using any classic sorting algorithm (e.g., quick sort algorithm),
rules R1 and R2 can only be consecutively applied finitely for each term $\tau_i$ due to the premises ${\tt sort}_{\geq_p} (m_1, \cdots, m_k)\neq (m_1, \cdots, m_k)$
and ${\tt sort}_{\geq_l} (\alpha_1, \cdots, \alpha_k)\neq (\alpha_1, \cdots, \alpha_k)$ in rules R1 and R2, respectively.

Second, rules R10,  R11 and  R12 can only be applied finitely in the reduction $\pi$, as these rules always push the addition operator $\oplus$ toward the root
of the syntax tree of the term $\tau_i$ when one of them is applied onto a term $\tau_i$, while the other rules
either eliminate or reorder the addition operator $\oplus$. 


Lastly, rules R3--9 and  R13 can only be applied finitely in the reduction $\pi$, as
these rules reduce the size of the term by $1$ when one of them is applied onto a term $\tau_i$ while
the rules R10--12 that  increase the size of the term can only be applied finitely. 

Hence, the reduction $\pi$ is finite indicating that the TRS $\mathcal{R}$ is terminating.
\qed
\end{proof}

By Lemmas~\ref{lemma:normalform} and~\ref{lemma:terminating}, any term $\tau\in T[\Sigma,V]$ can be rewritten to a normal form that must be a polynomial.

\begin{theorem}\label{thm:terminating}
Let $\mathcal{R}=(\Sigma, V, \geq_s, \Delta)$ be the TRS of a program $\mathcal{P}$. For any term $\tau\in T[\Sigma,V]$,
a polynomial $\tau'\in T[\Sigma,V]$ can be computed such that  $\tau \Rrightarrow\tau'$.
\end{theorem}

\begin{remark}
Besides the termination of a TRS, confluence is another important property of a TRS, where a TRS is confluent if any given term $\tau \in T[\Sigma,V]$ can be rewritten to two distinct terms $\tau_1$ and $\tau_2$,
then the terms $\tau_1$ and $\tau_2$ can be reduced to a common term.
While we conjecture that the TRS $\mathcal{R}$ associated with the given program is indeed confluent which may be shown by its local confluence~\cite{newman1942theories},
we do not strive to prove its confluence, as it is irrelevant to the problem considered in the current work.
\end{remark}

\section{Algorithmic Verification}\label{sec:algs}
In this section, we first present an algorithm for computing normal forms,  
then show how to compute the affine constant for an affine transformation, and finally
propose an algorithm for solving the verification problem. 

\subsection{Term Normalization Algorithm}

We provide the function \TermNorm (cf.\ Alg.~\ref{alg:termnorm}) which applies the rewriting rules in a particular order 
aiming for better efficiency.
Fix a TRS $\mathcal{R}=(\Sigma, V, \geq_s, \Delta)$, a term $\tau\in T[\Sigma,V]$
and a mapping $\lambda$ that provides required affine constants $\lambda(f)$.
 \TermNorm$(\mathcal{R},\tau,\lambda)$ returns a normal form $\tau'$ of $\tau$,
i.e., $\tau\Rrightarrow \tau'$. 

\begin{algorithm}[t]
\caption{Term Normalization}\label{alg:termnorm}
\SetKwProg{Fn}{Function}{:}{}
\Fn{\TermNorm{$\mathcal{R}$, $\tau$, $\lambda$}}{
    Rewrite $\tau$ by iteratively applying rules R$_3$--R$_{13}$ until no more update;\; \label{alg:termnorm-1}
    $\tau'\leftarrow {\tt sort}(\tau)$ by iteratively applying rule R$_{2};$\; \label{alg:termnorm-2}
    $\tau'\leftarrow {\tt sort}(\tau')$ by iteratively applying rule R$_{1}$;\; \label{alg:termnorm-3}
    Rewrite $\tau'$ by iteratively applying rules R$_3$, R$_{6}$, R$_{7}$ until no more update;\;  \label{alg:termnorm-4}
    \Return{$\tau'$}\;
}
\end{algorithm}

\TermNorm first  applies rules R3--R13 to rewrite the term $\tau$ (line~\ref{alg:termnorm-1}),
resulting in a polynomial which does not have $0$ as a factor or monomial (due to rules R4--R7),
or $1$ as a factor in a monomial  unless the monomial itself is $1$ (due to  rules R$_8$ and R$_9$).
Next, it recursively sorts all the factors and monomial involved in the polynomial
from the innermost sub-terms (lines~\ref{alg:termnorm-2} and  \ref{alg:termnorm-3}).
Sorting factors and monomials 
will place the same monomials at adjacent positions.
Finally, rules R3 and R6--R7 are further 
applied to simplify the polynomial (line~\ref{alg:termnorm-4}),
where consecutive syntactically equivalent monomials 
will be rewritten to $0$ by rule R3,
which may further enable rules R6--R7.
Obviously, the final term $\tau'$ is a normal form of the input $\tau$, although its size may be exponential in that of $\tau$. 

\begin{lemma}\label{lemma:alg1}
 \TermNorm$(\mathcal{R},\tau,\lambda)$  returns a normal form $\tau'$ of $\tau$.\qed
\end{lemma}

\subsection{Computing Affine Constants} \label{sect:const}
The function {\AffConst} 
in Alg.\ref{alg:affconst} computes
the associated affine constant for an affine transformation $f$.
It first 
sorts all affine transformations 
in a topological order based on the call graph $G$ (lines~\ref{alg:affconst-forbegin}--\ref{alg:affconst-forend}).
If  $f$ is \emph{only} declared in $\mathcal{P}$,
as mentioned previously, we assumed it is linear, thus $0$ is assigned to $\lambda(f)$ (line~\ref{alg:affconst-2}).
Otherwise, it extracts the input $x$ of $f$  
and computes its output $\xi(x)$ via symbolic execution (line~\ref{alg:affconst-4}),
where $\xi(x)$ is treated as $f(x)$.
We remark that during symbolic execution, we adopt a lazy strategy for inlining invoked affine transformations in $f$ to reduce
the size of $\xi(x)$.
Thus, $\xi(x)$ may contain affine transformations.

\begin{algorithm}[t]
\caption{Computing Affine Constants}\label{alg:affconst}
\SetKwProg{Fn}{Function}{:}{}
\Fn{\AffConst{$\mathcal{P}, \mathcal{R}, G$}}{
    \ForEach{affine transformation $f$ in a topological order of call graph $G$}{ \label{alg:affconst-forbegin}
        \If{$f$ is \emph{only} declared in $\mathcal{P}$}{$\lambda(f)\leftarrow0$;}  \label{alg:affconst-2}
        \Else{
            $x\leftarrow$input of $f$;\;  \label{alg:affconst-3}
            $\xi(x)\leftarrow {\tt symbolicExecution}(f)$;\;    \label{alg:affconst-4}
            $\tau\leftarrow \xi(x)[x\mapsto x\oplus y]\oplus \xi(x) \oplus \xi(x)[x\mapsto y]$;\;  \label{alg:affconst-5}
            \While{{\tt True}}{     \label{alg:affconst-whilebegin}
                $\tau\leftarrow \TermNorm(\mathcal{R},\tau,\lambda)$;\;   \label{alg:affconst-6}
                \If{$\tau$ is some constant $c$}{$\lambda(f)\leftarrow c$; \Break;}   \label{alg:affconst-7}
                \ElseIf{$g$ is defined in $\mathcal{P}$ but has not been inlined in $\tau$}{   \label{alg:affconst-8}
                    Inline $g$ in $\tau$; \Continue;\;   \label{alg:affconst-9}
                }
                \ElseIf{$\tau$ does not contain any uninterpreted function}{
                    $v_1,u_1,v_2,u_2\leftarrow$random values from $\mathbb{F}$ s.t. $v_1\neq v_2\vee u_1\neq u_2$;\;  \label{alg:affconst-91}
                    \If{$\tau[x\mapsto v_1,y\mapsto u_1]\neq \tau[x\mapsto v_2,y\mapsto u_2]$}{{\tt Emit}($f$ is not affine) and {\tt Abort};} \label{alg:affconst-92}
                }
                \If{{\tt SMTSolver}$(\forall x.\forall y. \tau=c)$={\tt SAT}}{   \label{alg:affconst-10}
                    $\lambda(f)\leftarrow$extract $c$ from the model; \Break;   \label{alg:affconst-11}
                }
                \lElse{{\tt Emit}($f$ may not be affine) and {\tt Abort};}   \label{alg:affconst-12}
            }    \label{alg:affconst-whileend}
        }
    }   \label{alg:affconst-forend}
    \Return{$\lambda$;}\;
}
\end{algorithm} 

Recall that $c$ is the affine constant of $f$  iff $\forall x, y\in \mathbb{F}. f(x \oplus y) = f(x) \oplus f(y) \oplus c$ holds.
Thus, we create the term $\tau=\xi(x)[x\mapsto x\oplus y]\oplus \xi(x) \oplus \xi(x)[x\mapsto y]$ (line~\ref{alg:affconst-4}), where
$e[a\mapsto b]$ denotes the substitution of $a$ with $b$ in $e$.
Obviously, the term $\tau$ is equivalent to some constant $c$ iff $c$ is the affine constant of $f$.

The while-loop (lines~\ref{alg:affconst-whilebegin}--\ref{alg:affconst-whileend}) evaluates $\tau$.
First, it rewrites $\tau$ to a normal form (line~\ref{alg:affconst-6}) by invoking \TermNorm in Alg.\ref{alg:termnorm}. If the normal form
is some constant $c$, then $c$ is the affine constant of $f$.
Otherwise, \AffConst repeatedly inlines each affine transformation $g$ that is defined in $P$ but has not been inlined in $\tau$ (lines~\ref{alg:affconst-8} and \ref{alg:affconst-9})
and rewrites the term $\tau$ to a normal form until either the normal form is some constant $c$ or no affine transformation can be inlined.
If the normal form is still not a constant, 
$\tau$ is evaluated using random input values. Clearly, if $\tau$ is evaluated to two distinct values (line~\ref{alg:affconst-92}), $f$ is not  affine.
Otherwise, we check the satisfiability of
the constraint $\forall x,y. \tau=c$ via an SMT solver in bitvector theory (line~\ref{alg:affconst-10}), where declared but undefined affine transformations are treated
as uninterpreted functions provided with their affine properties.
If $\forall x, y. \tau=c$ is satisfiable, we extract the affine constant $c$
from its model (line~\ref{alg:affconst-11}).
Otherwise, we emit an error and then abort  (line~\ref{alg:affconst-12}), indicating that the affine constant of $f$ cannot be computed.
Since the satisfiability problem module bitvector theory is decidable, we can conclude that $f$ is \emph{not} affine if  $\forall x.\forall y. \tau=c$ is unsatisfiable
and no uninterpreted function is involved in $\tau$.

\begin{lemma}\label{lemma:alg2}
Assume an affine transformation $f$ in $\mathcal{P}$.
If \AffConst$(\mathcal{P},\mathcal{R},G)$ in Alg.~\ref{alg:affconst} returns a mapping $\lambda$, then
$\lambda(f)$ is the affine constant of $f$.\qed 
\end{lemma}

\subsection{Verification Algorithm}

The verification problem is solved by the function \Verifier$(\mathcal{P})$ in Alg.~\ref{alg:correctness},
which checks if $f_{\masked}\cong f_{\origin}$,
for each procedure $f$ defined in $\mathcal{P}$.
 It first preprocesses the given program $\mathcal{P}$ by inlining all the procedures, unrolling all the loops and
eliminating all the branches (line~\ref{alg:correctness-1}).
Then, it computes the corresponding TRS $\mathcal{R}$, call graph $G$ and affine constants as the mapping $\lambda$, respectively (line~\ref{alg:correctness-2}).
Next, it iteratively checks if $f_{\masked}\cong f_{\origin}$,
for each procedure $f$ defined in $\mathcal{P}$ (lines~\ref{alg:correctness-forbegin}--\ref{alg:correctness-forend}).

For each procedure $f$, it first extracts the inputs $a^1, \cdots, a^m$ of $f_{\origin}$ that are scalar variables (line~\ref{alg:correctness-3})
and input encodings $\myvec{a}^1, \cdots, \myvec{a}^m$ of  $f_{\masked}$ that are 
vectors of variables (line~\ref{alg:correctness-4}).
Then, it computes the output $\xi(a^1, \cdots, a^m)$ of $f_{\origin}$ via symbolic execution, which yields 
an expression in terms of $a^1, \cdots, a^m$ and affine transformations (line~\ref{alg:correctness-5}).
Similarly, it computes the output $\myvec{\xi'}(\myvec{a}^1,   \cdots, \myvec{a}^m)$ of $f_{\masked}$ via symbolic execution, i.e.,  a tuple
of expressions in terms of the entries of the input encodings $\myvec{a}^1, \cdots, \myvec{a}^m$, random variables and affine transformations (line~\ref{alg:correctness-6}).

Recall that $f_{\masked}\cong f_{\origin}$ iff for all $a^1, \cdots, a^m, r_1,\cdots,r_h \in \mathbb{F}$ and for all $\myvec{a}^1,   \cdots, \myvec{a}^m \in \mathbb{F}^{d+1}$, the following constraint holds (cf. Definition~\ref{def:problem}):
\begin{center}
$\big(
	\bigwedge_{i \in [1,m]}\ a^i = \bigoplus_{j \in [0,d]} {\myvec{a}}^i_j
	\big)
	\rightarrow
	\big(f_{\origin}(a^1,\cdots, a^m) = \bigoplus_{i \in [0,d]} f_{\masked}(\myvec{a}^1, \cdots, \myvec{a}^m_i)\big)$
\end{center}
where $r_1,\cdots,r_h$ are all the random variables used in $f_{\masked}$.
Thus, it creates the term $\tau=\xi(a^1, \cdots, a^m)[a^1\mapsto \bigoplus\myvec{a}^1,\cdots, a^m\mapsto \bigoplus\myvec{a}^m]\oplus \bigoplus \myvec{\xi'}(\myvec{a}^1,   \cdots, \myvec{a}^m)$ (line~\ref{alg:correctness-7}), where $a^i\mapsto \bigoplus\myvec{a}^i$ is the substitution of $a^i$ with the term $\bigoplus\myvec{a}^i$ in the expression $\xi(a^1, \cdots, a^m)$.
Obviously, $\tau$ is equivalent to $0$ iff $f_{\masked}\cong f_{\origin}$.
 
\begin{algorithm}[t]
\caption{Verification Algorithm}\label{alg:correctness}
\SetKwProg{Fn}{Function}{:}{}
\Fn{\Verifier{$\mathcal{P}$}}{
   Inline all the procedures, unroll loops and eliminate branches in  $\mathcal{P}$;\;  \label{alg:correctness-1}
   $\mathcal{R}\leftarrow {\tt buildTRS}(\mathcal{P})$; 
    $G\leftarrow {\tt buildCallGraph}(\mathcal{P})$; 
    $\lambda\leftarrow \AffConst(\mathcal{P}, \mathcal{R}, G)$;\; \label{alg:correctness-2}
    \ForEach{procedure $f$ defined in $\mathcal{P}$}{ \label{alg:correctness-forbegin}
        Let $a^1, \cdots, a^m$ be the inputs of $f_{\origin}$;\;   \label{alg:correctness-3}
        Let $\myvec{a}^1,   \cdots, \myvec{a}^m$ be the input encodings of $f_{\masked}$;\;   \label{alg:correctness-4}
        $\xi(a^1, \cdots, a^m)\leftarrow {\tt symbolicExecution}(f_{\origin})$;\;    \label{alg:correctness-5}
        $\myvec{\xi'}(\myvec{a}^1,   \cdots, \myvec{a}^m)\leftarrow {\tt symbolicExecution}(f_{\masked})$;\;    \label{alg:correctness-6}
        $\tau\leftarrow \xi(a^1, \cdots, a^m)[a^1\mapsto \bigoplus\myvec{a}^1,\cdots, a^m\mapsto \bigoplus\myvec{a}^m]\oplus \bigoplus \myvec{\xi'}(\myvec{a}^1,   \cdots, \myvec{a}^m)$;\;    \label{alg:correctness-7}
        \While{{\tt True}}{     \label{alg:correctness-whilebegin}
            $\tau\leftarrow \TermNorm(\mathcal{R},\tau,\lambda)$\;     \label{alg:correctness-8}
            \If{$\tau$ is some constant $c$}{
                \lIf{$c=0$}{{\tt Emit}($f$ is correct);  \Break;}     \label{alg:correctness-9}
                \lElse{{\tt Emit}($f$ is incorrect);  \Break;}       \label{alg:correctness-10}
            }
            \ElseIf{$g$ is defined in $\mathcal{P}$ but has not been inlined in $\tau$}{    \label{alg:correctness-11}
                Inline $g$ in $\tau$; \Continue;\;     \label{alg:correctness-12}
            }
            \ElseIf{$\tau$ does not contain any uninterpreted function}{
                $\myvec{v}^1,   \cdots, \myvec{v}^m\leftarrow$random values from $\mathbb{F}^{d+1}$;\; \label{alg:correctness-120}
                \If{$\tau[\myvec{a}^1\mapsto \myvec{v}^1,\cdots, \myvec{a}^m\mapsto \myvec{v}^m]\neq 0$}{{\tt Emit}($f$ is incorrect); \Break;} \label{alg:correctness-121}
            }
            \If{{\tt SMTSolver}$(\tau\neq 0)$={\tt UNSAT}}{   \label{alg:correctness-13}  {{\tt Emit}($f$ is correct);} \Break;   \label{alg:correctness-14}
            }
            \lElse{{\tt Emit}($f$ may be incorrect); \Break;}    \label{alg:correctness-15}
        }    \label{alg:correctness-whileend}

    }  \label{alg:correctness-forend}
}
\end{algorithm} 

To check if  $\tau$ is equivalent to $0$, similar to computing affine constants
in Alg.~\ref{alg:affconst}, the algorithm repeatedly rewrites the term $\tau$ to a normal form by invoking  \TermNorm in Alg.~\ref{alg:termnorm}
until either the 
conclusion is drawn or no affine transformation can be inlined (lines~\ref{alg:correctness-whilebegin}--\ref{alg:correctness-whileend}).
We declare that $f$ is correct if the normal form is $0$ (line~\ref{alg:correctness-9})
and incorrect if it is a non-zero constant (line~\ref{alg:correctness-10}).
If the normal form is \emph{not} a constant,
we repeatedly inline  affine transformation $g$ defined in $P$ which has not been inlined in $\tau$
and re-check the term $\tau$. 

If 
there is no definite answer after inlining all the affine transformations,
$\tau$ is evaluated using random input values.
$f$ is \emph{incorrect} if $\tau$ is non-zero (line~\ref{alg:correctness-121}).
Otherwise, we check the satisfiability of
the constraint $\tau\neq 0$ via an SMT solver in bitvector theory (line~\ref{alg:correctness-13}).
If  $\tau\neq 0$ is unsatisfiable, then $f$ is \emph{correct}.
Otherwise 
we can conclude that $f$ is \emph{incorrect} if  
no uninterpreted function is involved in $\tau$, but in other cases it is not conclusive.

\begin{theorem}\label{thm:verif}
Assume a procedure $f$ in $\mathcal{P}$.
If \Verifier$(\mathcal{P})$ emits ``\emph{$f$ is correct}'', then $f_{\masked}\cong f_{\origin}$;
  if \Verifier$(\mathcal{P})$ emits ``\emph{$f$ is incorrect}'' or ``\emph{$f$ may be incorrect}'' with no uninterpreted function
involved in its final term $\tau$, then $f_{\masked}\not\cong f_{\origin}$. \qed
\end{theorem}


\subsection{Implementation Remarks}
To implement the algorithms, we use the total order $\geq_s$ on $V\uplus \Sigma$ 
where all the constants are smaller than the variables, which are in turn smaller than the affine transformations. The order of constants is the standard one on integers, and the order of variables (affine transformations)
uses lexicographic order.

In terms of data structure, each term is primarily stored by a directed acyclic graph, 
allowing us to represent and rewrite common sub-terms in an optimised way.
Once a (sub-)term becomes a polynomial during term rewriting, it is stored
as a sorted nested list w.r.t. the monomial order $\geq_p$, where
each monomial is also stored as a sorted list w.r.t. the factor order $\geq_l$. Moreover,
the factor of the form $\alpha^k$ in a monomial is stored by a pair $(\alpha,k)$.

We also adopted two strategies: 
(i) By Fermat's little theorem~\cite{vinogradov2016elements}, 
$x^{2^n - 1} = 1$ for any $x\in \mathbb{GF}(2^n)$. Hence  each $k$ in $(\alpha,k)$ can be simplified to $k \mod (2^n-1)$.
(ii) By rule R12, a term $f(\tau_1\oplus \cdots \oplus  \tau_k)$ can be directly rewritten to
$f(\tau_1)\oplus \cdots \oplus  (\tau_k)$ if $k$ is odd, and $f(\tau_1)\oplus \cdots \oplus  f(\tau_k)\oplus c$ if $k$ is even,
where $c$ is the affine constant associated with the affine transformation $f$.

\section{Evaluation}\label{sec:exp} 
We implement our approach as a tool {\tool} for verifying masked programs in LLVM IR,
based on the LLVM framework. We first evaluate {\tool} for computing
affine constants (i.e., Alg.~\ref{alg:affconst}), correctness verification, and scalability w.r.t.\ the masking order (i.e., Alg.~\ref{alg:correctness})
on benchmarks using the ISW scheme. To show the generality of our approach,
{\tool} is then used to verify benchmarks using glitch-resistant Boolean masking schemes and lattice-based public-key cryptography.
All experiments are conducted on a machine with Linux kernel 5.10, Intel i7 10700 CPU (4.8 GHz, 8 cores, 16 threads) and 40 GB memory.
Milliseconds (ms) and seconds (s) are used as the time units in our experiments.
 
\subsection{Evaluation for Computing Affine Constants} 
To evaluate Alg.~\ref{alg:affconst},
we compare with a pure SMT-based approach which directly checks 
$\exists c.\forall x, y\in \mathbb{F}. f(x \oplus y) = f(x) \oplus f(y) \oplus c$ using Z3~\cite{moura2008z3}, CVC5~\cite{BarbosaBBKLMMMN22} and Boolector~\cite{brummayer2009boolector},
by implementing $\oplus$ and $\otimes$ in bit-vector theory, where $\otimes$ is achieved via the Russian peasant method~\cite{bowden1912russian}.
Technically, SMT solvers only deal with satisfiability, but they usually can eliminate the universal quantifiers in this case, as $x,y$ are over a finite field.
In particular, in our experiment,
Z3 is configured with default (i.e. \texttt{(check-sat)}), simplify (i.e. \texttt{(check-sat-using (then simplify smt))}) and bit-blast (i.e. \texttt{(check-sat-using (then bit-blast smt))}), denoted
by Z3-d, Z3-s and Z3-b, respectively.
%
We focus on the following functions: $\texttt{exp}i(x)=x^i$ for $i\in\{2,4,8,16\}$;
$\texttt{rotl}i(x)$ for $i\in\{1,2,3,4\}$ that left rotates $x$ by $i$ bits;
$\texttt{af}(x) = \texttt{rotl1}(x) \oplus \texttt{rotl2}(x) \oplus \texttt{rotl3}(x) \oplus \texttt{rotl4}(x) \oplus 99$ used in AES S-Box;
$\texttt{L1}(x)=7x^2\oplus14x^4\oplus7x^8$, $\texttt{L3}(x)=7x\oplus12x^2\oplus12x^4\oplus9x^8$, $\texttt{L5}(x)=10x\oplus9x^2$ and $\texttt{L7}(x)=4x\oplus13x^2\oplus13x^4\oplus14x^8$ used in PRESENT S-Box over $\mathbb{GF}(16)=\mathbb{GF}(2)[X]/(X^4+X+1)$~\cite{bogdanov2007present,carlet2012higher};
$\texttt{f1}(x)=x^3$, $\texttt{f2}(x)=x^2\oplus x\oplus 1$, $\texttt{f3}(x)=x\oplus x^5$ and
$\texttt{f4}(x)=\texttt{af}(\texttt{exp2}(x))$ over $\mathbb{GF}(2^8)$.%

The results are reported in Table~\ref{tab:exp-affine}, where the 2nd--8th rows show the execution time
 and the last row shows the affine constants if they exist  otherwise \ding{55}.
\begin{table}[t]
	\centering
	\caption{Results of computing affine constants, where $\dag$ means Alg.~\ref{alg:affconst} needs SMT solving, $\ddag$ means affineness is disproved via testing,
\ding{55} means nonaffineness, and Alg.~\ref{alg:affconst}+B means Alg.~\ref{alg:affconst}+Boolector.}\vspace{-1mm}
\setlength{\tabcolsep}{1.2pt}
	\label{tab:exp-affine}
    \scalebox{0.66}{
	\begin{tabular}{|c||c|c|c|c|c|c|c|c|c|c|c|c|c|c|c|c|c|}
        \hline
        Tool & \texttt{exp2} & \texttt{exp4} & \texttt{exp8} & \texttt{exp16}
        & \texttt{rotl1} & \texttt{rotl2} & \texttt{rotl3} & \texttt{rotl4}
        & \texttt{af} & \texttt{L1} & \texttt{L3} & \texttt{L5} & \texttt{L7}
        & \texttt{f1} & \texttt{f2} & \texttt{f3} & \texttt{f4}\\
        \hline \hline
        Alg.~\ref{alg:affconst}+Z3-d & 3ms & 3ms & 3ms & 3ms & 18ms$^{\dag}$ & 18ms$^{\dag}$ & 18ms$^{\dag}$ & 18ms$^{\dag}$ & 21ms$^{\dag}$
        & 3ms & 3ms & 3ms & 3ms & 3ms$^{\ddag}$ & 3ms & 3ms$^{\ddag}$ & 21ms$^{\dag}$\\
         Alg.~\ref{alg:affconst}+Z3-b & 3ms & 3ms & 3ms & 3ms & 15ms$^{\dag}$ & 16ms$^{\dag}$ & 15ms$^{\dag}$ & 15ms$^{\dag}$ & 20ms$^{\dag}$
        & 3ms & 3ms & 3ms & 3ms & 3ms$^{\ddag}$ & 3ms & 3ms$^{\ddag}$ & 20ms$^{\dag}$\\
         Alg.~\ref{alg:affconst}+B & 3ms & 3ms & 3ms & 3ms & 8ms$^{\dag}$ & 8ms$^{\dag}$ & 8ms$^{\dag}$ & 8ms$^{\dag}$ & 13ms$^{\dag}$
        & 3ms & 3ms & 3ms & 3ms & 3ms$^{\ddag}$ & 3ms & 3ms$^{\ddag}$ & 14ms$^{\dag}$\\ \hline
         Z3-d & 181ms & 333ms & 316ms & 521ms & 14ms & 14ms & 14ms & 14ms & 16ms
        & 113ms & 213ms & 73ms & 194ms & 33ms & 249ms & 38ms & 7.5s\\
        Z3-s & 180ms & 373ms & 452ms & 528ms & 12ms & 12ms & 12ms & 12ms & 15ms
        & 158ms & 202ms & 194ms & 213ms & 28ms & 252ms & 35ms & 7.6s\\
         Z3-b & 15ms & 16ms & 18ms & 20ms & 12ms & 12ms & 12ms & 12ms & 79ms
        & 45ms & 42ms & 21ms & 82ms & 17ms & 22ms & 24ms & 60ms\\ \hline
         Boolector & 15ms & 18ms & 12ms & 17ms & 5ms & 5ms & 6ms & 5ms & 71ms & 25ms & 34ms & 27ms & 78ms & 14ms & 15ms & 17ms & 67ms\\ \hline
         CVC5 & 8.4s & 20.3s & 44.4s & 18.6s & 5ms & 5ms & 5ms & 5ms & 113ms
        & 158.4s & 263.4s & 43.7s & 214.9s & 92ms & 10.3s & 2.3s & 10.4s\\ \hline
        \hline
     Result & 0 & 0 & 0 & 0 & 0 & 0 & 0 & 0 & 99
        & 0 & 0 & 0 & 0 & \ding{55} & 1 &  \ding{55}  & 99\\
        \hline
	\end{tabular}
    }
\end{table}
We observe
that Alg.~\ref{alg:affconst} significantly outperforms the SMT-based approach on most cases for all the SMT solvers,
except for \texttt{rotl$i$} and \texttt{af} (It is not surprising, as
they use operations rather than $\oplus$ and $\otimes$, thus SMT solving is required).
The term rewriting system is often able to
compute affine constants \emph{solely} (e.g., $\texttt{exp}i$ and  $\texttt{L}i$), 
and SMT solving is required \emph{only} for computing the affine constants of $\texttt{rotl}i$. By comparing the results of Alg.~\ref{alg:affconst}+Z3-b vs. Z3-b and
Alg.~\ref{alg:affconst}+B vs. Boolector on
\texttt{af}, we observe that term rewriting is essential as checking normal form---instead of the original constraint---reduces the cost of SMT solving.

\subsection{Evaluation for Correctness Verification}\label{sec:comparsion}

To evaluate Alg.~\ref{alg:correctness}, we compare it with a pure SMT-based approach with SMT solvers Z3, CVC5 and Boolector.
We also consider several promising general-purpose software verifiers
SMACK (with Boogie and Corral engines), SeaHorn, CPAChecker
and Symbiotic, and one cryptography-specific verifier CryptoLine (with SMT and CAS solvers), where
the verification problem is expressed using \textcolor{officegreen}{assume} and \textcolor{officegreen}{assert}  statements.
Those verifiers are configured in two ways: (1) recommended ones in the manual/paper or used in the competition, and (2) by trials of different configurations and selecting the optimal one.
Specifically:
\begin{itemize}
  \item CryptoLine (commit \texttt{7e237a9}). Both solvers SMT and CAS are used;
  \item SMACK v2.8.0. integer-encoding: bit-vector, verifier: corral/boogie (both used), solver: Z3/CVC4 (Z3 used), static-unroll: on, unroll: 99;
  \item SEAHORN v0.1.0 RC3 (commit \texttt{e712712}). pipeline: bpf, arch: m64, inline: on, track: mem, bmc: none/mono/path (mono used), crab: on/off (off used);
\item CPAChecker v2.1.1. default.properties with cbmc: on/off (on used);
\item Symbiotic v8.0.0. officially-provided SV-COMP configuration with exit-on-error: on.
\end{itemize}

The benchmark comprises five different masked programs \texttt{sec\_mult}  for finite-field multiplication  over $\mathbb{GF}(2^8)$
by varying masking order $d=0,1,2,3$, where
the $d=0$ means the program is unmasked. We note that \texttt{sec\_mult} in \cite{barthe2017parallel} is only available for masking order $d\ge 2$.

The results are shown in Table~\ref{tab:exp-secmult}.
We can observe that {\tool} is significantly more efficient than the others,
and is able to prove all the cases  using our term rewriting system \emph{solely} (i.e., without random testing or SMT solving).
With the increase of masking order $d$, almost all the other tools failed. 
Both CryptoLine (with the CAS solver) and CPAChecker fail to verify any of the cases due to the non-linear operations involved in \texttt{sec\_mult}.
SMACK with Corral engine produces two false positives (marked by $\natural$ in Table~\ref{tab:exp-secmult}).
These results suggest that dedicated verification approaches are required for proving the correctness
of masked programs.
\begin{table}[t]
	\centering
	\caption{Results on various \texttt{sec\_mult},
        where T.O. means time out (20 minutes),
        N/A means that UNKNOWN result,
        and $\natural$ means that verification result is \emph{incorrect}.}\vspace{-1mm}
	\label{tab:exp-secmult}\setlength{\tabcolsep}{1pt}
    \scalebox{0.715}{
	\begin{tabular}{|c|c|c|c|c|c|c|c|c|c|c|c|c|c|c|c|}
        \hline
       Order  & \multirow{2}{*}{Ref.} & \multirow{2}{*}{Alg.~\ref{alg:correctness}} & \multicolumn{3}{c|}{Z3} & \multirow{2}{*}{Boolector}
       & \multirow{2}{*}{CVC5} & \multicolumn{2}{c|}{CryptoLine} & \multicolumn{2}{c|}{SMACK} & \multirow{2}{*}{SeaHorn} & \multirow{2}{*}{CPAChecker} & \multirow{2}{*}{Symbiotic}\\
        \cline{4-6}\cline{9-12}
         \multicolumn{1}{|c|}{$d$} & \multicolumn{1}{c|}{} & & default & simplify & bit-blast & & & SMT & CAS & Boogie & Corral & & &\\
        \hline
        \multirow{4}{*}{0} & \cite{rivain2010provably} & 17ms & 29ms & 27ms & 42ms &25ms & 29ms & 39ms & N/A & 29s & 66s & 132ms & T.O. & 870s\\
        & \cite{belaid2016randomness} & 20ms & 31ms & 31ms & 45ms & 28ms & 33ms & 35ms & N/A & 46s & 144s & 128ms & T.O. & 899s\\
        & \cite{gross2017reconciling} & 21ms & 33ms & 31ms & 46ms & 29ms & 33ms & 32ms & N/A & 23s & 43s & 127ms & T.O. & 872s\\
        & \cite{cassiers2020trivially} & 18ms & 30ms & 28ms & 25ms & 26ms & 31ms & 32ms & N/A & 17s & 56s & 130ms & T.O. & 876s\\
        \hline
        \multirow{4}{*}{1} & \cite{rivain2010provably} & 18ms & 298ms & 299ms & 391s & 3.8s & T.O. & 469ms & N/A & T.O. & T.O. & 13s & T.O. & T.O.\\
        & \cite{belaid2016randomness} & 20ms & 299ms & 299ms & 1049s & 1.9s & T.O. & 582ms & N/A & T.O. & T.O. & 13s & T.O. & T.O.\\
        & \cite{gross2017reconciling} & 24ms & 295ms & 295ms & 1199s & 1.8s & T.O. & 951ms & N/A & T.O. & T.O. & 14s & T.O. & T.O.\\
        & \cite{cassiers2020trivially} & 20ms & 1180s & 921s & T.O. & 7.7s & T.O. & 21s & N/A & T.O. & T.O. & T.O. & T.O. & T.O.\\
        \hline
        \multirow{4}{*}{2} & \cite{rivain2010provably} & 20ms & 4.1s & 4.2s & T.O. & T.O. & T.O. & T.O. & N/A & T.O. & T.O. & T.O. & T.O. & T.O.\\
        & \cite{belaid2016randomness} & 22ms & 4.2s & 4.4s & T.O. & T.O. & T.O. & T.O. & N/A & T.O. & T.O. & T.O. & T.O. & T.O.\\
        & \cite{barthe2017parallel} & 30ms & 4.2s & 4.1s & T.O. & T.O. & T.O. & T.O. & N/A & T.O. & 26s$^\natural$ & T.O. & T.O. & T.O.\\
        & \cite{gross2017reconciling} & 29ms & 4.2s & 4.2s & T.O. & T.O. & T.O. & T.O. & N/A & T.O. & T.O. & T.O. & T.O. & T.O.\\
        & \cite{cassiers2020trivially} & 22ms & T.O. & T.O. & T.O. & T.O. & T.O. & T.O. & N/A & T.O. & T.O. & T.O. & T.O. & T.O.\\
        \hline
        \multirow{4}{*}{3} & \cite{rivain2010provably} & 21ms & T.O. & T.O. & T.O. & T.O. & T.O. & T.O. & N/A & T.O. & T.O. & T.O. & T.O. & T.O.\\
        & \cite{belaid2016randomness} & 26ms & T.O. & T.O. & T.O. & T.O. & T.O. & T.O. & N/A & T.O. & T.O. & T.O. & T.O. & T.O.\\
        & \cite{barthe2017parallel} & 27ms & T.O. & T.O. & T.O. & T.O. & T.O. & T.O. & N/A & T.O. & 1059s$^\natural$ & T.O. & T.O. & T.O.\\
        & \cite{gross2017reconciling} & 29ms & T.O. & T.O. & T.O. & T.O. & T.O. & T.O. & N/A & T.O. & T.O. & T.O. & T.O. & T.O.\\
        & \cite{cassiers2020trivially} & 24ms & T.O. & T.O. & T.O. & T.O. & T.O. & T.O. & N/A & T.O. & T.O. & T.O. & T.O. & T.O.\\
        \hline
	\end{tabular}
    }
\end{table}

\subsection{Scalability of {\tool}}\label{sec:scale}
To evaluate the scalability of {\tool}, we verify different versions of 
\texttt{sec\_mult}  
and masked procedures \texttt{sec\_aes\_sbox} (resp.\ \texttt{sec\_present\_sbox}) of S-Boxes
used in AES~\cite{rivain2010provably} (resp.\ PRESENT~\cite{carlet2012higher})
with varying masking order $d$.
Since it is known that  \texttt{refresh\_masks} in~\cite{rivain2010provably} is vulnerable when $d\geq 4$~\cite{CPRR13},
a fixed version \texttt{RefreshM}~\cite{BBDFGSZ16}
is used in all the S-Boxes (except
that when \texttt{sec\_mult} is taken from \cite{barthe2017parallel} its own version  is used).
We note that \texttt{sec\_present\_sbox} uses the affine transformations \texttt{L1}, \texttt{L3}, \texttt{L5}, \texttt{L7}, \texttt{exp2} and \texttt{exp4},
while \texttt{sec\_aes\_sbox}  uses the affine transformations \texttt{af}, \texttt{exp2}, \texttt{exp4} and \texttt{exp16}.

The results are reported in Table \ref{tab:exp-scale}.
All those benchmarks are proved using our term
rewriting system solely except for the three incorrect ones marked by $\natural$.
{\tool} scales up to masking order of 100 or even 200 for
\texttt{sec\_mult}, which is remarkable. 
{\tool} also scales up to masking order of 30 or even 40 for \texttt{sec\_present\_sbox}.
However, it is less scalable on  \texttt{sec\_aes\_sbox}, as
it computes the multiplicative inverse $x^{254}$ on shares, and the size of the term encoding the equivalence problem explodes
with the increase of the masking order.
Furthermore, to better demonstrate the effectiveness of our term writing system in dealing with complicated procedures, we first use Algorithm~\ref{alg:affconst} 
to derive affine constants on \texttt{sec\_aes\_sbox} with ISW~\cite{rivain2010provably} and then directly apply SMT solvers to solve the correctness constraints obtained at Line~\ref{alg:correctness-7} of Algorithm~\ref{alg:correctness}. It takes about 1 second to obtain the result on the first-order masking, while 
fails to obtain the result within 20 minutes on the second-order masking.

\begin{table}[t]
	\centering
	\caption{Results on \texttt{sec\_mult} and S-Boxes, where T.O. means time out (20 minutes),
    and $\natural$ means that the program is \emph{incorrect}.	\label{tab:exp-scale}}\vspace{-1mm}\setlength{\tabcolsep}{1pt}
    \scalebox{0.75}{
	\begin{tabular}{|c|c|c|c|c|c|c|c|c|c|c|c|c|c|c|c|c|c|}
        \hline
       \multirow{2}{*}{\diagbox{Ref.}{$d$}}     & \multicolumn{6}{c|}{\texttt{sec\_mult}} & \multicolumn{7}{c|}{\texttt{sec\_present\_sbox}} & \multicolumn{4}{c|}{\texttt{sec\_aes\_sbox}}\\
        \cline{2-18}
         & 5 & 10 & 20 & 50 & 100 & 200 
        & 1 & 2 & 5 & 10 & 20 & 30 & 40
        & 1 & 2 & 4 & 5 \\ 
        \hline
       ISW \cite{rivain2010provably} & 23ms & 33ms & 84ms & 1.0s & 15s & 545s
        & 44ms & 51ms & 93ms & 535ms & 14s & 118s & T.O.
        & 87ms & 234ms & 25s & 160s \\ 
       ISW \cite{belaid2016randomness} & 26ms & 44ms & 100ms & 712ms & 7.3s & 212s
        & 54ms & 63ms & 110ms & 673ms & 17s & 163s & T.O.
        & 108ms & 265ms & 23s & 142s \\ 
      ISW  \cite{barthe2017parallel} & 36ms$^\natural$ & 49ms & 109ms & 601ms & 3.2s & 18s
        & -- & 86ms & 142ms$^\natural$ & 237ms & 841ms & 2.4s & 5.3s
        & -- & 559ms & 9.7s & 142s$^\natural$ \\ 
       ISW \cite{gross2017reconciling} & 34ms & 50ms & 98ms & 518ms & 3.1s & 19s
        & 67ms & 91ms & 137ms & 700ms & 20s & 173s & T.O.
        & 140ms & 571ms & 63s & T.O. \\ 
       ISW \cite{cassiers2020trivially} & 30ms & 109ms & 224ms & 5.0s & 152s & T.O.
        & 51ms & 61ms & 113ms & 354ms & 2.4s & 9.7s & 29s
        & 133ms & 269ms & 13s & 68s \\ 
        \hline
	\end{tabular}
    }
\end{table}

A highlight of our findings is that {\tool} reports that \texttt{sec\_mult} from~\cite{barthe2017parallel} and the S-boxes based on this version are incorrect when $d=5$.
After a careful analysis, we found that indeed
it is incorrect for any $d \equiv 1 \mod 4$ (i.e., 5, 9, 13, etc).
This is because \cite{barthe2017parallel} parallelizes the multiplication over the entire encodings (i.e., tuples of shares) while the parallelized computation depends on the value of $d \mod 4$.
When the reminder is $1$, the error occurs.

\subsection{Evaluation for more Boolean Masking Schemes}

To demonstrate the applicability of {\tool} on a wider range of Boolean masking schemes, we further consider
glitch-resistant Boolean masking schemes: HPC1, HPC2~\cite{cassiers2020hardware}, DOM~\cite{gross2016domain} and CMS~\cite{reparaz2015consolidating}.
We implement the finite-field multiplication \texttt{sec\_mult} using those masking schemes, as well as
masked versions of AES S-box and PRESENT S-box.
We note that our implementation of DOM \texttt{sec\_mult} is derived from \cite{cassiers2020hardware},
and we only implement the 2nd-order CMS \texttt{sec\_mult} due to the difficulty of implementation.
All other experimental settings are the same as in Section \ref{sec:scale}.

\begin{table}[t]
	\centering
	\caption{Results on \texttt{sec\_mult} and S-Boxes for HPC, DOM and CMS.	
    \label{tab:exp-hpc-dom-cms}}\vspace{-1mm}
    \setlength{\tabcolsep}{1.6pt}
    \scalebox{0.81}{
	\begin{tabular}{|c|c|c|c|c|c|c|c|c|c|c|c|c|c|c|c|c|}
        \hline
       \multirow{2}{*}{\diagbox{Ref.}{$d$}}     & \multicolumn{6}{c|}{\texttt{sec\_mult}} & \multicolumn{5}{c|}{\texttt{sec\_present\_sbox}} & \multicolumn{5}{c|}{\texttt{sec\_aes\_sbox}}\\
        \cline{2-17}
         & 0 & 1 & 2 & 3 & 4 & 5 
        & 1 & 2 & 3 & 4 & 5
        & 1 & 2 & 3 & 4 & 5 \\ 
        \hline
        HPC1~\cite{cassiers2020hardware} & 28ms & 30ms & 32ms & 35ms & 39ms & 42ms
        & 63ms & 72ms & 84ms & 98ms & 117ms
        & 104ms & 254ms & 1.8s & 13s & 67s \\ 
        HPC2~\cite{cassiers2020hardware} & 23ms & 25ms & 26ms & 28ms & 31ms & 33ms
        & 57ms & 66ms & 75ms & 92ms & 110ms
        & 92ms & 244ms & 1.9s & 13s & 65s \\ 
        DOM~\cite{gross2016domain} & 24ms & 24ms & 25ms & 26ms & 28ms & 29ms
        & 52ms & 60ms & 67ms & 77ms & 90ms
        & 80ms & 223ms & 1.8s & 12s & 66s \\ 
        CMS~\cite{reparaz2015consolidating} & -- & -- & 24ms & -- & -- & -- 
        & -- & 53ms & -- & -- & --
        & -- & 211ms & -- & -- & -- \\ 
        \hline
	\end{tabular}
    }
\end{table}

The results are shown in Table \ref{tab:exp-hpc-dom-cms}.
Our term rewriting system \emph{solely} is able to efficiently prove the correctness of finite-field multiplication \texttt{sec\_mult}, masked versions of AES S-box and PRESENT S-box
using the glitch-resistant Boolean masking schemes HPC1, HPC2, DOM and CMS. The verification cost of those benchmarks
is similar to that of benchmarks using the ISW scheme, demonstrating the applicability of {\tool} for various Boolean masking schemes.


\subsection{Evaluation for Arithmetic/Boolean Masking Conversions}\label{Bitslice}

\begin{table}[t]
	\centering\setlength{\tabcolsep}{1pt}
	\caption{Results on \texttt{sec\_add}, \texttt{sec\_add\_modp} and \texttt{sec\_a2b}~\cite{bronchain2022bitslicing}, where T.O. means time out (20 minutes).
    \label{tab:exp-sec-add}}\vspace{-1mm} 
    \scalebox{0.66}{
	\begin{tabular}{|c|c|c|c|c|c|c|c|c|c|c|c|c|c|c|c|c|c|c|c|c|c|c|c|c|c|c|c|c|}
        \hline
      \multirow{2}{*}{\diagbox{$d$}{$k$}}   & \multicolumn{7}{c|}{\texttt{sec\_add}} & \multicolumn{6}{c|}{\texttt{sec\_add\_modp}}
      & \multicolumn{7}{c|}{\texttt{sec\_a2b}}\\
        \cline{2-21}
        & 2 & 3 & 4 & 6 & 8 & 12 & 16 &
        2 & 3 & 4 & 6 & 8 & 12 &
        2 & 3 & 4 & 6 & 8 & 12 & 16\\  \hline
       1  & 34ms & 38ms & 42ms & 51ms & 61ms & 83ms & 109ms &
       97ms & 248ms & 805ms & 7.5s & 44s & 623s &
       41ms & 48ms & 55ms & 70ms & 87ms & 121ms & 156ms\\ \hline
       2  & 35ms & 40ms & 45ms & 55ms & 65ms & 91ms & 124ms &
       111ms & 331ms & 1.1s & 11s & 67s & 936s &
       58ms & 74ms & 93ms & 134ms & 199ms & 523ms & 1.5s\\ \hline
       3  & 36ms & 42ms & 47ms & 58ms & 71ms & 100ms & 139ms &
       127ms & 417ms & 1.5s & 15s & 89s & T.O. &
       73ms & 93ms & 118ms & 182ms & 293ms & 927ms & 3.0s\\ \hline
       4 & 38ms & 44ms & 50ms & 62ms & 76ms & 109ms & 155ms &
       144ms & 506ms & 1.9s & 18s & 112s & T.O. &
       93ms & 130ms & 190ms & 676ms & 3.3s & 49s & 366s\\ \hline
       5  & 39ms & 45ms & 51ms & 66ms & 81ms & 118ms & 168ms &
        160ms & 586ms & 2.2s & 22s & 136s & T.O. &
        109ms & 159ms & 256ms & 1.1s & 6.5s & 100s & 746s\\
        \hline
	\end{tabular}
    }
\end{table}

To demonstrate a wider applicability of {\tool} 
other than masked implementations of symmetric cryptography, we further evaluate {\tool} on three key non-linear building blocks for bitsliced, masked implementations of lattice-based post-quantum key encapsulation mechanisms (KEMs~\cite{bronchain2022bitslicing}).
Note that KEMs are a class of encryption techniques designed to secure symmetric cryptographic key material for transmission using asymmetric (public-key) cryptography.
We implement the Boolean masked addition modulo $2^k$ (\texttt{sec\_add}), Boolean masked addition modulo $p$ (\texttt{sec\_add\_modp}) and the arithmetic-to-Boolean masking conversion modulo $2^k$ (\texttt{sec\_a2b}) for various bit-width $k$ and masking order $d$, where $p$ is the largest prime number less than $2^k$.
Note that some bitwise operations (e.g., circular shift) are expressed by affine transformations,
and the modulo addition is implemented by the simulation algorithm~\cite{bronchain2022bitslicing} in our implementations.

The results are reported in Table \ref{tab:exp-sec-add}.
{\tool} is able to efficiently prove the correctness of these functions with various masking orders ($d$) and bit-width ($k$), using the term rewriting system \emph{solely}.
With the increase of the bit-width $k$ (resp.\ masking order $d$), the verification cost increases more quickly
for \texttt{sec\_add\_modp} (resp. \texttt{sec\_a2b}) than for \texttt{sec\_add}. This is
because \texttt{sec\_add\_modp} with bit-width $k$ invokes \texttt{sec\_add} three times, two of which have the bit-width $k+1$, and the number of calls to \texttt{sec\_add} in \texttt{sec\_a2b} increases with the masking order $d$
though using the same bit-width as \texttt{sec\_a2b}. These results demonstrate the applicability of {\tool} for asymmetric cryptography.


\section{Conclusion}\label{sec:concl}
We have proposed a term rewriting based approach to proving
functional equivalence between masked cryptographic programs and their original
unmasked algorithms over $\mathbb{GF}(2^n)$.
Based on this approach, we have developed a tool {\tool} and carried out extensive experiments on various benchmarks. Our evaluation confirms the effectiveness, efficiency and applicability of our approach.

For future work, it would be interesting to further investigate the theoretical properties of the term rewriting system. 
Moreover, we believe the term rewriting approach extended with more operations may have a greater potential in verifying more general cryptographic programs, e.g., those from the standard software library such as OpenSSL.


\bibliographystyle{splncs04}
\bibliography{refs}



\end{document}